\documentclass[twocolumn,prl,aps,superscriptaddress]{revtex4}
\usepackage{amsmath,amssymb,latexsym,revsymb,color}
\usepackage[pdftex]{graphicx}
\pdfoutput=1

\def\duzomniejsze{<\kern-.7mm<}
\def\duzowieksze{>\kern-.7mm>}

\def\textbf#1{{\bf #1}}
\def\beq{\begin{equation}}
\def\eeq{\end{equation}}
\def\be{\begin{equation}}
\def\ee{\end{equation}}
\def\ben{\begin{eqnarray}}
\def\een{\end{eqnarray}}
\def\beqa{\begin{eqnarray}}
\def\eeqa{\end{eqnarray}}
\def\eea{\end{array}}
\def\bea{\begin{array}}
\newcommand{\bei}{\begin{itemize}}
\newcommand{\eei}{\end{itemize}}
\newcommand{\bee}{\begin{enumerate}}
\newcommand{\eee}{\end{enumerate}}

\newcommand{\tr}{\operatorname{Tr}}

\def\id{{\rm I}}

\def\>{\rangle}
\def\<{\langle}

\def\ot{\otimes}

\newcommand{\ket}[1]{| #1 \rangle}
\newcommand{\bra}[1]{\langle #1 |}
\newcommand{\proj}[1]{\ket{#1}\bra{#1}}

\newtheorem{lemma}{Lemma}

\newtheorem{theorem}[lemma]{Theorem}
\newtheorem{definition}[lemma]{Definition}

\def\12{{\textstyle \frac{1}{2}}}

\newcounter{protoline}
\newlength{\boxwidth}
\newlength{\bigboxwidth}

\newtheorem{proto_body}{Protocol}

\newenvironment{protocol}[1]{ 
  \setcounter{protoline}{0}
  \begin{minipage}{\boxwidth}
    
    \begin{proto_body}[ #1 ]
      \ \\[0mm] 
      \begin{description} }{\end{description}\end{proto_body}
\end{minipage}}
\newenvironment{protocol_cont}[0]{ 

\begin{minipage}{\boxwidth}
\addtocounter{proto_body}{-1}

\begin{proto_body}
\ \\ 
\begin{description} }{\end{description}\end{proto_body}
\end{minipage}}

\newenvironment{proof}[1][Proof]{\noindent \textbf{{#1~} }}{\qed}
\def\squareforqed{\hbox{\rlap{$\sqcap$}$\sqcup$}}
\def\qed{\ifmmode\squareforqed\else{\unskip\nobreak\hfil
\penalty50\hskip1em\null\nobreak\hfil\squareforqed
\parfillskip=0pt\finalhyphendemerits=0\endgraf}\fi}

\newcommand{\rateq}{Q}
\newcommand{\ratec}{C}
\newcommand{\sent}{\alpha}

\newcommand{\enviro}{\alpha'}
\newcommand{\mutind}{I_\Lambda}
\newcommand{\mutindss}{I_{ss}}
\newcommand{\spacey}{\,\,\,}
\begin{document}

\title{The quantum one-time pad in the presence of an eavesdropper}

\begin{abstract}

\end{abstract}
\author{Fernando G.S.L. Brand\~ao}
\affiliation{Departamento de F\'isica, Universidade Federal de Minas Gerais,
     Belo Horizonte, Caixa Postal 702, 30123-970, MG, Brazil}
\author{Jonathan Oppenheim}
\affiliation{Department of Applied Mathematics and Theoretical Physics, University of Cambridge U.K.}
\begin{abstract}
A classical one-time pad allows two parties to send private messages over
a public classical channel -- an eavesdropper who intercepts the communication learns nothing
about the message.  A quantum one-time pad is a shared quantum state
which allows two parties to send
private messages or private quantum states over a public quantum channel.  
If the eavesdropper intercepts
the quantum communication she learns nothing about the message.  In the 
classical case, a one-time pad can be created using shared and partially
private correlations. 
Here we consider the quantum case in the presence of an eavesdropper, 
and find the single letter
 formula for the rate at which the two parties can send messages using a quantum one-time pad.
\end{abstract}
\maketitle

\noindent
\textit{Introduction.} If two parties wish to send private messages over a public channel,
then they need to share a {\it one-time pad} or {\it key}  -- 
perfectly correlated and private strings which are as long as the messages
they want to send.  Often, the strings they share are not perfectly correlated
or not completely secure e.g.\ if produced through a channel subject to wire-tapping.  However, they can perform a protocol over the public
channel to reconcile
the errors in their strings, and amplify the privacy, so that they share
a shorter string which is perfectly correlated and private. Given access to many independent realizations of some distribution $P_{\text{XYZ}}$ shared between the two parties, Alice (X) and Bob (Y), and an eavesdropper Eve (Z), the rate $\ratec(P_{XYZ})$ at which Alice can send private messages to Bob was derived in \cite{Ashlwede-Csiszar1993}, based on a celebrated result due to Wyner and Csiszar \& Korner~\cite{Wyner-wiretap,CsiszarKorner}. It reads
\beq
\ratec(P_{XYZ})=\sup_{X \rightarrow U \rightarrow V}I(V:Y|U)-I(V:Z|U),
\eeq
with the conditional mutual information $I(V:Y|U):=H(VU)+H(YU)-H(VYU)-H(U)$, the Shannon entropy $H(X):=-\sum_x P_{X=x}\log P_{X=x}$ and the supremum taken over the Markov chain $X\rightarrow U\rightarrow V$. 

The quantum analog of this is three parties, Alice Bob and Eve, who instead of sharing a classical distribution, share a quantum state $\psi_{ABE}$. Alice then wishes to send private messages or private quantum states to Bob over a quantum public channel i.e.\ an insecure quantum channel where the eavesdropper might intercept the sent states. The question of how many private messages can be sent using a shared state was posed and answered by Schumacher and Westmoreland~\cite{schumacher2006quantum} in the case where initially the eavesdropper is uncorrelated with the two parties ($\psi_{ABE}=\psi_{AB}\otimes\psi_E$), and the sent messages are classical. They proved that the rate of classical private messages which can be sent is given by the quantum mutual information $I(A:B):=S(A)+S(B)-S(AB)$, with $S(A)=-\tr\rho_A\log\rho_A$ the von Neumann entropy.

Here, we consider the general case where the two parties want to protect themselves against an eavesdropper who might be correlated with their state. We also extend the result to the case where the parties wish to send encrypted quantum states to each other, i.e. any input state $\psi_{K}$ in dimension $\log{d}$ is encrypted so that during transmission it is indistinguishable from the maximally mixed state ($\id/\log{d}$). This makes the scenario a more fully quantum version of the classical situation.  We will find, in surprising analogy with the classical case, that the rate $\rateq$ that Alice can send encrypted quantum states to Bob using the state $\psi_{ABE}$ is 
\begin{equation} \label{optimalratesingleletter}
 \rateq(\psi_{ABE})=\sup_{A\rightarrow a\alpha}\frac{1}{2}
(I(a:B|\alpha)-I(a:E|\alpha)),
\end{equation}
with the conditional mutual information 
$I(a:B|\alpha):=S(a\alpha)+S(B\alpha)-S(aB\alpha)-S(\alpha)$ and 
the supremum
taken over channels which maps $\psi_A$ to $\rho_{a\alpha}$. Using simple entropic identities, one sees that the right hand side of Eq. (\ref{optimalratesingleletter}) is equivalent to
$\frac{1}{2}(I(a:B\alpha)-I(a:E\alpha))$, a quantity which
has made an early appearance in Ref. \cite{SSW2008} as the distillable entanglement assisted by {\it symmetric-side channels}. 
Note that this optimisation is over single copies of the state $\psi_{ABE}$ making the result of Equation (\ref{optimalratesingleletter}) {\it single-letter}.
This is rare in quantum information theory, where usually the solutions are intractable, requiring optimisation over arbitrary many copies of the state.

\vspace{0.1 cm}

\noindent
\textit{Statement of the problem.} The scenario is as follows: Alice and Bob share many copies of a quantum system in a (generally mixed) state $\psi_{AB}$ and since we want to protect against an arbitrary eavesdropper, we should imagine that Eve might have any state such that 
$\tr_E\ket{\psi}_{ABE}\bra{\psi}_{ABE}=\psi_{AB}$, i.e. the eavesdropper might hold a {\it purification} of Alice and Bob's state. Alice is given a message, either classical or quantum, which she should communicate to Bob. She is able to implement arbitrary quantum operations on her share $\psi_A^{\otimes n}$ of the state and any local ancillas, and she then sends a quantum system in state $\rho_\sent$ to Bob down an insecure quantum channel, which might be intercepted by Eve. In the case where Eve intercepts $\rho_\sent$, she should learn an arbitrarily small amount of information about the message. 
In the case where Bob receives the state, he should be able
to recover the message with probability converging to one in the limit 
of large $n$. More formally:

\begin{definition}[private state transfer]
\label{def:merging}  
  
Consider the {\it message} state $\Psi_{KR}$ shared between the sender Alice and a reference. Let Alice, Bob and Eve share the state $\ket{\psi_{ABE}}^{\otimes n}$ and have further registers
  $a,\sent$ and $b$ for Alice and Bob, respectively. Consider Alice's local operation (a completely positive trace preserving map)
  ${\cal M_A}:A K \longrightarrow a\alpha$ and Bob's local operation
  ${\cal M_B}:B \sent  \longrightarrow b$. Then
  a \emph{private state transfer protocol} for $\Psi_{KR}$ has error $\delta$ and security
parameter $\epsilon$, if
  \beq
    \label{eq:def-merg}
    \| \rho_{bR} - \Psi_{KR}  \|_1
                                                                  \leq \delta,\eeq
and 
\beq
    \label{eq:def-privacy}
    \| \tilde \rho_{RE\sent } -\tilde \rho_R\ot \tilde\rho_{E\sent} \|_1
                                                                  \leq \epsilon,
\eeq
where
\begin{equation}
\rho_{RKa\alpha BE} := {\cal M}_A(\Psi_{KR}\otimes \psi_{ABE}^{\otimes n}),
\end{equation}
and
\begin{equation}
\tilde \rho_{RKa bE} := {\cal M}_B \circ{\cal M}_A(\Psi_{KR}\otimes \psi_{ABE}^{\otimes n}).
\end{equation}
\end{definition}

For classical messages we let $\Psi_{KR} = \frac{1}{d}\sum_k \ket{kk}\bra{kk}_{KR}$ and define the optimal rate $C(\rho_{AB})$ as the ratio of $\log(d)$ per $n$, for the largest $d$ for which a private state transfer protocol is possible, with negligible error for asymptotic large $n$. 

For the optimal rate of quantum messages, in turn, we set $\ket{\Psi_{KR}} = \frac{1}{\sqrt{d}}\sum_{k}\ket{k, k}_{KR}$ and define $Q(\rho_{AB})$ as the asymptotic optimal ratio of $\log(d)/n$, over all private state transfer protocols. 
\vspace{0.1 cm}

\noindent
\textit{Schumacher-Westmoreland scheme.} To prove Eq. (\ref{optimalratesingleletter}), we will make use of the result from \cite{schumacher2006quantum} for the one-time-pad in the case where the message is classical and the state $\rho_{AB}$ shared by Alice and Bob is not correlated with Eve. The main point of the argument is the construction of a set of quantum operations $\{{\cal\ E}_{k, n} \}$ on Alice's system and a probability distribution $\{ p_{k, n} \}$ such that in the limit of large $n$,      
\begin{equation} \label{decoding}
\frac{1}{n}\chi(\{ p_{k, n}, {\cal E}_{k, n}\otimes \id_{B}(\psi_{AB}^{\otimes n})\}) \rightarrow I(A:B)_{\rho},
\end{equation}
and
\begin{equation} \label{privacy}
\frac{1}{n}\chi(\{ p_{k, n}, {\cal E}_{k, n}(\psi_{A}^{\otimes n})\} \rightarrow 0,
\end{equation}
where $\chi(\{ q_k, \sigma_k\}) := S(\sum_k q_k \sigma_k) - \sum_k q_k S(\sigma_k)$ is the Holevo information \cite{Holevo98}. By the HSW theorem \cite{SchumacherW97} Alice can then send secret classical messages to Bob at a rate $I(A:B)$ by applying one of the ${\cal\ E}_{k, n}$ operations to her part of the state and sending it down the insecure channel. Eq. (\ref{decoding}) guarantees that Bob is able to decode Alice's message in the case the channel is not tampered, while Eq. (\ref{privacy}) ensures that Eve does not learn anything from the message being sent by intercepting the channel. 
 
\vspace{0.1 cm}

\noindent
\textit{Mutual independence.} A natural quantity which will arise in our discussion is the so-called {\it mutual independence} $I_\Lambda$~\cite{horodecki2009quantum}, which we now define. Consider some sequence of maps $\Lambda^{(n)}$, from a restricted class of operations $\Lambda$, applied to subsystem $AB$ with the property that
\begin{equation}
\rho^{(n)}_{ABE} := \Lambda^{(n)}\otimes \id_{E}(\psi_{ABE}^{\otimes n})
\end{equation}
is such that
\begin{equation}
\Vert \rho^{(n)}_{ABE} - \rho^{(n)}_{AB} \otimes \rho^{(n)}_{E} \Vert_1 \rightarrow 0.
\label{eq:MIcriteria}
\end{equation}
Then
\begin{definition}[mutual independence] \label{midef}
Given a state $\psi_{AB}$, consider a protocol from a class of operations $\Lambda$ for extracting mutual independence ${\cal P} = \Lambda^{(n)}$. Define the rate
\begin{equation}
R({\cal P}, \rho_{AB}) := \liminf_{n \rightarrow \infty} \frac{1}{2} I(A : B)_{\Lambda^{n}(\psi_{AB}^{\otimes n})}.
\end{equation}
Then we define the mutual independence rate of $\psi_{AB}$ as
\begin{equation}
I_\Lambda(\rho_{AB}) := \sup_{\cal P} R({\cal P}, \psi_{AB}).
\end{equation}
\end{definition}

The quantity $\mutind$ can be thought of as the rate of private mutual information
that can be extracted from a state under the class of operations $\Lambda$. As an immediate consequence of Schumacher-Westmoreland construction and Definition \ref{midef}, we find that $C(\psi_{AB})$ is lower bounded by $I_{\text{1-LOCC}}(\psi_{AB})$, where 1-LOCC is the class of local operations assisted by one-way classical communication. It turns out, perhaps surprisingly, that one-way LOCC is not the right class of operations to be considered here! 

As we show, the rate of private messages that can be sent is given by $\mutindss(\psi_{AB})$, the mutual independence when $\Lambda$ is the class of local operations assisted by a \textit{symmetric-side channel}. This is a channel given
by an isometry followed by partial trace 
$\psi_A \rightarrow \tr_E\rho_{BE}$  such that $\rho_{BE}$ is unchanged after interchanging
system $E$ with system $B$.  In \cite{brandao2010public}, it is shown that
\beq \label{eq:ISS}
\mutindss(\psi_{AB})=\sup_{A\rightarrow a\alpha}\frac{1}{2}
(I(a:B|\alpha)-I(a:E|\alpha))
\eeq
where the supremum is taken over channels $A\rightarrow a\sent$.
%
In \cite{brandao2010public}, we prove as well that this same quantity is equal to a weaker variant of mutual independence, in which Eq. (\ref{eq:MIcriteria}) is replaced by the weaker criteria
\begin{equation}
\Vert \rho^{(n)}_{AE} - \rho^{(n)}_{A} \otimes \rho^{(n)}_{E} \Vert_1 \rightarrow 0.
\label{eq:WMIcriteria}
\end{equation}
\vspace{0.1 cm}

\noindent
\textit{Main result.} We now show
\begin{theorem}
\beq
\rateq(\psi_{AB})=\ratec(\psi_{AB})/2=\mutindss(\psi_{AB})
\eeq
\end{theorem}
\begin{proof}
We begin by considering $\ratec(\psi_{AB})$, i.e.\ Alice wishes to send Bob
a private classical message, and will then prove the result
for $\rateq(\psi_{AB})$. To see that
$
\mutindss(\psi_{AB}) \geq \ratec(\psi_{AB}) / 2, 
$
consider an optimal protocol for $\ratec(\psi_{AB})$, which can always be taken to be as follows: Alice applies the quantum operation ${\cal E}_{k, n} \otimes \id_{BE}$ with probability $p_{k, n}$, generating the ABE ensemble $\{ p_{k, n}, {\cal E}_{k, n}(\psi_{ABE}) \}$, with $\rho_\sent={\cal E}_{k, n}(\psi_{A})$ being sent to Bob, and $k$ the private message to be communicated. Then we have
\begin{equation}
\ratec(\psi_{AB}) = \lim_{n \rightarrow \infty} \frac{1}{n}\chi( p_{k, n}, {\cal E}_{k, n}\otimes \id_B(\psi_{AB})).
\end{equation}
Consider the state after Alice's optimal local operation
\begin{equation}
\rho_{KABE}^n := \sum_k p_{k, n} \ket{k}_K\bra{k} \otimes  \left({\cal E}_{k, n}\otimes \id_{BE}\right)(\psi_{ABE}) 
\end{equation}
Then, from  Eq. (\ref{eq:ISS}) we get
\begin{equation}
\mutindss(\psi_{AB})  \geq \frac{1}{2} \left(  I(K : B\sent)_{\rho} - I(K : E\sent)_{\rho} \right).
\end{equation}
But  $I(K : B\sent)_\rho = \chi( p_{k, n}, {\cal E}_{k, n} \otimes \id_B(\psi_{AB}) )$ and  $I(K : E\sent)_\rho^n \rightarrow 0$ with increasing $n$, since ${\cal E}_{k, n} \otimes \id_{E}(\psi_{AE})$ must satisfy Condition (\ref{eq:def-privacy}) and be asymptotically independent of $k$. Therefore we get
$
\mutindss(\psi_{AB}) \geq \ratec(\psi_{AB})/2. 
$

Next we need to show that 
$
\mutindss(\psi_{AB}) \leq C(\psi_{AB})/2.
$
First, suppose that on top of the insecure ideal quantum channel Alice and Bob
have access to a symmetric-side channel. Then they could distill 
$\mutindss(\psi_{AB})$ of mutual
independence, using the symmetric side-channel. They are now in the situation
considered by Schumacher and Westmoreland, who showed that in the case where
Alice and Bob are initially product with Eve, $\ratec(\psi_{AB})=I(A:B)$. Thus here we would get
$\ratec(\psi_{AB})=2 \mutindss(\psi_{AB})$
of secure classical communication. 

Of course in the setting we are considering, they do not have access to the symmetric
side-channel. However suppose Alice simulates locally the side-channel, 
sends the part that would go to Bob through the insecure quantum channel and traces out the part 
which would go to 
Eve. Then, on one hand, if Eve does not intercept the channel, Bob will get his share of what is send through the 
symmetric side-channel and they can distill at least $\mutindss(\psi_{AB})$ of weak mutual independence and achieve the rate $\ratec=2 \mutindss(\psi_{AB})$. 
I.e.\ if Eve doesn't get her share of the output $\alpha'$ of the symmetric side-channel Alice and Bob can not be in a worse position than if she did receive it. 
On the other hand, if Eve intercepts the state sent through 
the insecure channel, then this is the same state she would get in 
the case they were connected by a symmetric side-channel (because
what goes to Bob and Eve is symmetric), so Eve must still be decoupled from Alice's final state. This is so because Alice and Eve's state must be product in the end of the protocol for distilling mutual independence. Thus she gets no information about $\rho_K$.

This proves $\ratec=2 \mutindss(\psi_{AB})$. That $\rateq(\psi_{AB})=\ratec(\psi_{AB})/2$ comes from
the fact that instead of using the quantum one-time pad to send private messages, Alice and Bob could just as well use it to share a classical private key
$\sum \proj{kk}_{AB}/d^2$. This key can then be used to encrypt quantum states which can then be sent through the insecure quantum channel.  

It is known \cite{boykin-qe,mosca-qe,HaydenLSW03-approxrandom} that the amount of key required to encrypt a state of dimension $\log d$ is given by $2\log d$. In more detail, The procedure for encrypting a quantum state is for  Alice to perform randomizing unitaries $\sum _k \proj{k} \otimes U_k$ controlled on the classical key where $U_k$ is a complete set of unitaries acting on the state she wants to encrypt. Bob can then decrypt the quantum state by performing $U_k^\dagger$. E.g. to encrypt a qubit, Alice acts one of the four Pauli operators $\id,\sigma_x,\sigma_y,\sigma_z$ with the choice of which operator to act decided by two bits of key. 
\end{proof}
\vspace{0.2 cm}
Note that when we are using the key to
encrypt quantum states, we can modify the protocol slightly to include an authentication step~\cite{debbie-qvc,dan-auth} so that if at some later point,  
Bob is allowed at least one bit of backwards communication, the key can be recycled~\cite{debbie-qvc,oh-recycling}
and used to encrypt more quantum states.  The bit of back-communication
is required to signal to Alice that the protocol succeeded (i.e. that Eve didn't disturb the sent states too much) and is not part of the orignal scenario
considered here.  However, in such a case,
one can prove that the one-time pad
can be recycled in the case where we are using it to send quantum states~\cite{oh-recycling}!
\vspace{0.1 cm}

\noindent
\textit{A direct protocol.} We can also construct a different protocol which encrypts quantum states directly using the one-time pad without first using it to create a classical key. This results in a saving of $\log{d}$ uses of the public quantum channel. 

Recall that to create a classical key, Alice applies ${\cal E}_k \otimes \id_{BE}(\psi_{ABE}^{\otimes n})$ conditioned on a random classical variable $k$. To encrypt a quantum state directly,
Alice applies ${\cal E}_k$ coherently, controlled on her half $K$ of the
entangled state $\psi_{KR}=\sum p_k \ket{k}_{R}\ket{k}_K$,
 i.e. she performs the operation $\sum \proj{k}_K\otimes {V}_k$, where $V_k$ is an isometric extension of the operation ${\cal E}_k$. This produces the total state $\ket{\Psi}=\sum p_k \ket{k}_{R}\ket{k}_K\ket{\psi^k}_{\alpha\alpha'BE}$
where $\rho_{\enviro}^k$ is the local environment produced under the action of map ${\cal E}_k$ and $\rho_\sent$ is its output.
Alice then sends $\rho_\sent$ to Bob, who can then coherently
decode $\rho^k_{\sent B}$ producing the state 
$\sum p_k \ket{k}_{R}\ket{k}_K\ket{k}_{B'}\ket{\psi^0}_{\sent\enviro BE}$. 
The protocol is thus far secure, because after tracing out system $K$, the state
$\rho_{R\sent E}$ is exactly the same as in the case of sending a 
classical message, and thus satisfies the privacy condition 
(\ref{eq:def-privacy}).

Since the state
$\sum p_k \ket{k}_{R}\ket{k}_K\ket{k}_{B'}$ has $S(K|B')=0$, Alice can {\it merge}~\cite{how-merge} her share ($K$) of the state to Bob by performing
a complete measurement in a random basis and communicating the result to Bob. In \cite{how-merge} it was shown that $S(K|B')$ 
is the amount of EPR  pairs that is needed to send Alice's share $K$ of 
$\ket\psi_{KB'R}$  by performing a measurement and if $S(K|B')=0$, then no additional EPR
pairs are needed. Alice's merging measurement
completely decouples the $K$ system from the reference, with the result that if
Alice sends the remainder of her systems to Bob, the state must have been transmitted.
She could also perform a measurement
in the Fourier basis and communicate the result. Since the measurement is complete,
the number of measurement outcomes is just $nH(K)$, and because
we wish Eve to learn no information about the state, Alice needs to use an additional
$n H(K)$ of the quantum one-time pad to encrypt the measurement result and send it. 

Alice's measurement result is independent of the final state (as in teleportation~\cite{teleportation}) so
we can do the measuring and sending coherently, which will result in $nH(K)$ 
EPR pairs being created \cite{family} in the case
where Eve does not interfere with the channel. However, these EPR pairs can only be used at some later time if Bob verifies that he received them using an authentication scheme
involving at least one bit of back-communication \cite{dan-auth}. Note that if $R$ is held
by Alice, both protocols for sending quantum states can also be used to create secure EPR pairs between Alice and Bob. 

The direct protocol for encrypting quantum states uses $\log d$ less uses of the channel than if we first create a classical key, and then send encrypted quantum states.  As a result, $\log{d}$ less bits of key is left 
over if we are allowed back communication at some later point in time
to recycle the key.  This is
in keeping with a fundamental law of privacy~\cite{oh-recycling} relating sent qubits $(\delta Q$), the change in the amount of shared key ($\delta K$), and messages 
sent ($\delta M$) (whether they be classical or quantum):  
\beq
\delta K \leq \delta Q - \delta M \spacey .
\eeq

It is also worth noting the connection between merging, and encryption of the quantum states in this case.  Encrypting the quantum
state means that Alice's share of $\ket\Psi_{KR}$ should be {\it decoupled} from the reference $R$ before being sent down the channel.  
At the same time, this decoupling of the reference from Alice's laboratory is the condition for Alice to succeed in sending her share \cite{how-merge,how-merge2}.  
\vspace{0.1 cm}

\noindent
\textit{Approximate encryption with half key.} As we have noted, the condition for decoupling system $K$ from the reference $R$ is that $2\log{d}$ unitaries are applied.  It turns out there is a weaker form of quantum state encryption, where only slightly more than $\log{d}$ bits of key are used~\cite{randomization}. In such a case, the protocol is secure in the sense that
if a measurement were to be performed on the reference system, then an eavesdropper would learn an arbitrary small amount about the measurement
result.  We say
that the level of security we
obtain is not {\it composable}~\cite{Ben-Or-Mayers,BHLMO}, meaning that if the reference system remains unmeasured, and the eavesdropper does not measure the parts of the quantum system she intercepted, then we may loose security if we use the encrypted state in another protocol. 

We can easily construct an encryption scheme of this sort, by adapting the first protocol we presented, so that instead of choosing a complete set of $2 \log d$ unitaries $U_k$ which act on the state we are encrypting, we choose  just
over $\log d$ unitaries at random from the Haar measure~\cite{aubrun2009almost}.  Such a set is called {\it randomizing} rather than  {\it completely randomizing}.  It is unclear whether the direct protocol can be adapted in some way for approximate encryption. This is because 
the protocol uses merging, and thus the state to be sent must be completely decoupled from the reference system.
\vspace{0.1 cm}

\noindent
\textit{Discussion. }There are essentially two ways we have used the quantum one-time pad.  One way is to use $\psi_{AB}$ to obtain a correlated and private key, and then use this key to encrypt messages (quantum or classical).  The second,  
is a generalisation of Schumacher and Westmoreland \cite{schumacher2006quantum} where the one-time is used directly to encrypt
the message.  This also holds true in the case of classical distributions.

Our results can also be applied to channel coding, where one has an authenticated noisy quantum channel, which produces
the state $\psi_{ABE},$ and a public quantum
channel.  Here we have just taken $\psi_{ABE}$ as a static resource, but we could
have just imagined that it was produced by a channel from Alice to Bob and Eve. 
This is perhaps closest to a quantum version of the Csiszar-Korner situation and
gives a physical application to the results of 
\cite{SSW2008,smith2008private,brandao2010public}, about state and channel capacities assisted by a symmetric-side channel.

We should thus think of a symmetric channel not as an exotic side-channel which can
be used in conjunction with a standard quantum channel.  Rather, results which make
use of a symmetric channel can be applied to the situation
where an eavesdropper might intercept the quantum
systems that are sent down an insecure channel.  This gives further motivation to
the notion of the public quantum channel as emphasised in \cite{brandao2010public}.

\vspace{0.1 cm}

\noindent
{\bf Acknowledgements.}
J.O. is supported by the Royal Society, and National Science Foundation under Grant No. PHY05-51164 during his stay at KITP.  F.B. is supported by a fellowship ``Conhecimento Novo'' from Funda\c{c}\~ao de Amparo a Pesquisa do Estado de Minas Gerais (FAPEMIG).

\bibliographystyle{apsrev}

\end{document}